\documentclass{amsart}
\usepackage{graphicx}
\usepackage{graphicx}
\usepackage{amssymb}
\usepackage{amsfonts}
\usepackage{amsmath}
\usepackage{amsthm}
\newtheorem{thm}{Theorem}[section]

\newtheorem{lem}[thm]{Lemma}
\newtheorem{rem}[thm]{Remark}

\numberwithin{equation}{section}

\begin{document}
\title{Error Estimates for Multinomial Approximations of American Options in Merton's Model}

 \author{Yan Dolinsky\\
 Department of Mathematics\\
Hebrew University, Jerusalem,
Israel }%

\address{
 Department of Mathematics\\
Hebrew University, Jerusalem, Isreel
 {e.mail: yann1@math.huji.ac.il}}

 \date{20.11.2009}

\begin{abstract}
We derive error estimates for multinomial approximations of American options in a multidimensional jump--diffusion Merton's model.
We assume that the payoffs are Markovian and satisfy Lipschitz type conditions. Error estimates for such type of approximations
were not obtained before. Our main tool is the strong approximations theorems
for i.i.d. random vectors which were obtained in \cite{S}.
For the multidimensional Black--Scholes model our results can be extended also to a general path dependent payoffs which satisfy Lipschitz type conditions.
For the case of multinomial approximations of American options for the Black--Scholes model our estimates are a significant improvement of those which were obtained in \cite{Ki3} (for game options in a more general setup).\end{abstract}
\maketitle
\markboth{Y.Dolinsly}{Error Estimates for American Options}
\renewcommand{\theequation}{\arabic{section}.\arabic{equation}}
\pagenumbering{arabic}

\section{Introduction}\label{sec:1}\setcounter{equation}{0}
This paper deals with multinomial approximations of American options arbitrage--free prices in the multidimensional jump--diffusion Merton's model with finite horizon. The Merton model is a generalization of the Black--Scholes model and it allows the stock to have jumps of compound Poisson type.
We consider a Markovian payoffs which satisfy Lipschitz type conditions and derive error estimates for an appropriate multinomial approximations.
In the multidimensional Black--Scholes model these error estimates can be also derived for general path dependent payoffs which satisfy
Lipschitz type conditions.

For American options in finite horizon Merton's model arbitrage--free prices can not be calculated explicitly.
Since multinomial models are active on a discrete set of times and defined on a discrete probability space, then arbitrage--free prices in these models
can be calculated efficiently by dynamical programming algorithm. Thus convergence results for multinomial approximations provides an efficient tool to evaluate arbitrage free prices in the Merton model.

Several papers dealt with multinomial approximations of American options in models with jumps (see, \cite{M} and \cite{MSS}).
In both of the papers the authors considered the case with one risky asset and used the weak convergence approach. The main tool that was used in the above papers is the stability results for Snell's envelopes under weak convergence which were obtained in \cite{MP}. In \cite{D} the weak convergence approach was used to show convergence results for multinomial approximations of game options (which were introduced in \cite{Ki1}) in the Merton model. The main disadvantage of the weak convergence approach is that this machinery can not provide, in principle, speed of convergence estimates. Thus no error estimates were obtained for multinomial approximations of American options in the Merton model.

Clearly, from practical point of view it is valuable to find estimates of the corresponding errors. In order to obtain error estimates we should consider all the market models on the same probability space, and so methods based on strong approximation theorems come into picture.
Strong approximation theorem allows to construct a probability space which contain all the markets models such that
the risky assets in the discrete time models will be "close" with respect to the $\sup$ norm to the risky assets in the continuous model.
Several authors applied strong invariance principles in order to obtain error estimates for American and game options in the one dimensional BS model
(see, \cite{LR},\cite{Ki2} and \cite{DK3}). In all of these papers the authors used the Skorokhod embedding tool of i.i.d. random variables into the one dimensional Brownian motion. This tool can not be applied for the multidimensional case. In \cite{Ki3} the author studied discrete time approximations of Dynkin's game values for the multidimensional Brownian motion. The main tool that was used there is strong approximation theorems which were developed in \cite{BP} and they work for sequences of random vectors with close characteristic functions.

In \cite{S} the author considered a new approach to strong approximations in the multidimensional case.
He showed that for a given sequence of i.i.d. random vectors $X^{(1)},X^{(2)}....$, and a random vector $Y$ which has the same expectation and covariance matrix as $X^{(1)}$, it is possible to construct a sequence of i.i.d. vectors $Y^{(1)},Y^{(2)}....$ such that $Y^{(1)}\sim Y$
and the normalized sums of the last sequence will be "close" to the normalized sums of the first sequence. Furthermore for any  $k<m$, $Y^{(k)}$ is independent of $X^{(m)}$. The above approach will be the main tool that we use in order to establish the results in this paper.
For Merton's model strong approximation theorems were not used before. An interesting question which is still open, is whether the method from \cite{S} can be applied for game options approximations.

Let us remark, that the estimates for the Brownian motion by means of normalized sums of independent random vectors obtained in \cite{S}
are much better than those which were obtained in \cite{BP}.

Main results of this paper are formulated in the next section where
we also introduce the notations that will be used. In Section 3 we
derive auxiliary lemmas that we use. In Section 4 we
complete the proof of main results of the paper.

\section{Preliminaries and main results}\label{sec:2}\setcounter{equation}{0}
Consider a complete probability space ($\Omega,
\mathcal{F},P)$ together with a standard $d$--dimensional
continuous in time Brownian motion \{$W(t)=(W_1(t),...,W_d(t))\}_{t=0}^\infty$, a Poisson process $\{N(t)\}_{t=0}^\infty$ with intensity $\lambda$ and independent of $W$, and a sequence of i.i.d. random vectors $\{U^{(i)}=(U^{(i)}_1,...,U^{(i)}_d)\}_{i=1}^\infty$ with values in $(-1,\infty)^d$, independent of $W$ and $N$. We also assume that for any $1\leq j\leq d$, $E|U^{(1)}_j|^2<\infty$ where $E$ denotes the expectation with respect to $P$.
A $d$--dimensional Merton's model with horizon $T<\infty$ consists of a savings account which given by
\begin{equation}\label{2.4}
\begin{split}
b(t)=b(0)\exp(rt),  \ b(0),r>0,  \ 0\leq t\leq T
\end{split}
\end{equation}
and of $d$ risky stocks ${\{S(t)=(S_1(t),...,S_d(t))\}}_{t=0}^T$ given by
\begin{equation}\label{2.5}
\begin{split}
S_i(t)=S_i(0)\exp((r+\mu_i-{\sum_{j=1}^d \sigma^2_{ij}}/{2})t
 +\sum_{j=1}^d\sigma_{ij} W_j(t))\prod_{j=1}^{N(t)}(1+U^{(j)}_i), \  1\leq i\leq d
\end{split}
\end{equation}
where $\sigma=(\sigma_{ij})_{1\leq i,j\leq d}$ is a nonsingular matrix, $S_i(0)>0$
and
without loss of generality we assume that
\begin{equation}\label{2.6}
\begin{split}
\mu_i=-\lambda EU^{(1)}_i, \  1\leq i\leq d.
\end{split}
\end{equation}
Consider an American option with the payoff process
\begin{equation}\label{2.7}
\begin{split}
Y(t)=F(S(t),t),  \  0\leq t\leq T
\end{split}
\end{equation}
where
$F:\mathbb{R}^d_{+}\times\mathbb{R}_{+}\rightarrow\mathbb{R}_{+}$ is a function such that
for some constant
$L\geq{1}$ and for any $t\geq{s}\geq{0}$ and
$\upsilon,\tilde{\upsilon}\in\mathbb{R}^d$,
\begin{equation}\label{2.7+}
|F(\upsilon,t)-F(\tilde{\upsilon},s)|\leq L \sum_{i=1}^d |\upsilon_i-\tilde{\upsilon}_i| +
L(t-s)(1+\sum_{i=1}^d |\upsilon_i|).
\end{equation}
Let $\mathcal{T}$ be the set of stopping times with respect to the natural filtration generated by $S$ (which satisfies the usual conditions) with values not exceeding $T$.
The equality (\ref{2.6})
guaranties that the probability measure $P$ is a martingale measure. Thus the term
\begin{equation}\label{2.8}
V=\sup_{\tau\in\mathcal{T}}E(\exp(-r\tau)Y(\tau))
\end{equation}
gives an arbitrage--free price for the American option.

For any $n\in\mathbb{N}$ define a sequence of i.i.d. random vectors ${\{U^{n,i}=(U^{n,i}_1,...,U^{n,i}_d)\}}_{i=1}^\infty$ by
\begin{equation}\label{2.8+}
\begin{split}
U^{n,i}_j=\sum_{k=1}^{M(n)} (\frac{k}{2}n^{-1/8}-1)\mathbb{I}_{\{\frac{k-1}{2}n^{-1/8}-1<U^{(i)}_j\leq \frac{k}{2}n^{-1/8}-1\}}, \ 1\leq j \leq d, \ i\in\mathbb{N}
\end{split}
\end{equation}
where $\mathbb{I}_Q=1$ if an event $Q$ occurs and $=0$ if not, and $M(n)\in\mathbb{N}$ satisfies
\begin{equation}\label{2.8++}
\sum_{j=1}^d E(U^{(1)}_j\mathbb{I}_{\{U^{(1)}_j>\frac{M(n)}{2}n^{-1/8}-1\}})<\frac{n^{-1/8}}{2}.
\end{equation}
Notice that
\begin{equation}\label{2.8+++}
E|U^{(i)}_j-U^{(n,i)}_j|<n^{-1/8}, \ i,n\in\mathbb{N}, \ 1\leq j\leq d.
\end{equation}
Next we describe the discrete time markets which we use to approximate the Merton model.
Let $A\in M_{d+1}(\mathbb{R})$ be an orthogonal matrix such that it last column equals to
$(\frac{1}{\sqrt{d+1}},...,\frac{1}{\sqrt{d+1}})$. Let $\Omega_{\xi}=\{1,2,...,d+1\}^\infty$ be the space of infinite sequences $\omega=(\omega_1,\omega_2,...)$; $\omega_i\in\{1,2,...,d+1\}$ with the product probability $P^{\xi}=\{\frac{1}{d+1},...,\frac{1}{d+1}\}^\infty$.
 Define a sequence of i.i.d. random vectors $\xi^{(1)},\xi^{(2)},...$ by
\begin{equation}\label{2.9}
\begin{split}
\xi^{(i)}(\omega)=\sqrt{d+1} (A_{\omega_i 1},A_{\omega_i 2}...,A_{\omega_i d}), \ i\in\mathbb{N}.
\end{split}
\end{equation}
Observe that
\begin{equation}\label{2.9+}
\begin{split}
E\xi^{(1)}=0 \  \mbox{and} \ E\xi^{(1)}_i\xi^{(1)}_j=1 \ \mbox{if} \  i=j \  \mbox{and} \ =0 \ \mbox{otherwise}.
\end{split}
\end{equation}
The probability space $(\Omega_{\xi},P^{\xi})$ was introduced in \cite{H}.
For any $n$ we extend $(\Omega_{\xi},P^{\xi})$ to a probability space $(\Omega_n,\mathcal{G}_n,P_n)$ such that it contains a three independent sequences of i.i.d. random vectors
$\{\xi^{(i)}\}_{i=1}^\infty$, $\{\rho^{n,k}\}_{k=1}^n$ and
$\{u^{n,k}\}_{k=1}^n$. The second sequence is a sequence of
Bernoulli random variables such that
$P_n\{\rho^{n,1}=1\}=1-\exp(-\lambda T/n)$ and the second sequence satisfies $u^{n,1}\sim U^{n,1}$.

For any $n\in\mathbb{N}$, $0\leq k\leq n$ and $1\leq i \leq d$ set $N^{n,k}=\sum_{m=1}^k \rho^{n,m}$ and
\begin{equation}\label{2.11}
S^{(n)}_i(t)=S_i(0)\exp(rkT/n)\prod_{m=1}^k (1+\sqrt\frac{T}{n}\sum_{j=1}^d \sigma_{ij}\xi^{(m)}_j)\frac{\prod_{j=1}^{N^{n,k}} (1+u^{n,j}_i)}{\big(1+(1-\exp(-\lambda T/n))E_n u^{n,1}_i\big)^k}
\end{equation}
where $\frac{kT}{n}\leq t< \frac{(k+1)T}{n}$ and $E_n$ denotes the expectation with respect to $P_n$.
Consider a multinomial $n$--step market which is active in the moments $0,\frac{T}{n},\frac{2T}{n},...,T$ and consists of a savings account
which given by (\ref{2.4})
and of $d$ risky stocks $S^{(n)}=(S^{(n)}_1,...,S^{(n)}_d)$ given by
(\ref{2.11}).
Next, we introduce an American option
with the payoff process
\begin{equation}\label{2.12}
Y^{(n)}(k)=F(S^{(n)}(kT/n),kT/n), \ 0\leq k\leq n.
\end{equation}
Let $\mathcal{T}_n$ be the set of stopping times with respect to the filtration $\{\sigma(S^{(n)}_0,S^{(n)}_{\frac{T}{n}},...,S^{(n)}_{\frac{kT}{n}})\}_{k=0}^n$
with values in $\{0,1,...,n\}$.
Observe that for any $n$, $P_n$ is a martingale measure for the $n$--step market. Thus
\begin{equation}\label{2.13}
V_n=\sup_{\tau\in\mathcal{T}_n} E_n(\exp(-\tau T/n)Y^{(n)}(\tau))
\end{equation}
is an arbitrage--free price of the $n$--step market.
The following theorem says that the arbitrage--free prices of the $n$--step markets converge to the arbitrage--free price of the Merton model and provides an estimates on the error terms.
\begin{thm}\label{thm2.1}
For any $\epsilon>0$ there exists a constant $C_{\epsilon}$ such that for any $n$
\begin{equation}\label{2.14}
|V-V_n|<C_{\epsilon}n^{\epsilon-\frac{1}{8}}.
\end{equation}
\end{thm}
\begin{rem}\label{rem2.1}
Theorem \ref{thm2.1} can be extended to a case where we have a finite number of Poisson clocks. Namely, consider a complete probability space ($\Omega,\mathcal{F},P)$ together with a standard $d$--dimensional
continuous in time Brownian motion $W$, $m$ independent Poisson processes
$N^{(1)},...,N^{(m)}$, which are independent of $W$, and for any $1\leq k\leq m$ a sequence of i.i.d.
random vectors $\{U^{k,i}\}_{i=1}^\infty$ with values in $(-1,\infty)^d$. We assume that the sequences are independent of each other and independent of $W$ and $N^{(1)},...,N^{(m)}$. We also assume $E|U^{k,1}_j|^2<\infty$, for $j\leq d$ and $k\leq m$. The risky assets are given by
\begin{equation}\label{2.15}
\begin{split}
S_i(t)=S_i(0)\exp((r+\mu_i-{\sum_{j=1}^d \sigma^2_{ij}}/{2})t
 +\sum_{j=1}^d\sigma_{ij} W_j(t))\prod_{k=1}^m\prod_{j=1}^{N^{(k)}(t)}(1+U^{k,j}_i)
\end{split}
\end{equation}
where $\mu_i=-\lambda \sum_{k=1}^m EU^{k,1}_i$, $i\leq d$. For an analogical multinomial models (to those that we used for one Poisson process) we can prove a similar result to Theorem \ref{thm2.1}. The proof for this case can be done in a similar way to the proof of Theorem \ref{thm2.1} and by using the same ideas. For simplicity we provide the proof only for the case with one Poisson process.
\end{rem}
\begin{rem}\label{rem2.2}
By using the Cauchy-Schwarz inequality
and the Chebyshev's inequality it follows that we can set $M(n)=[cn^{1/4}]$ for some constant $c$ ($[x]$ is the integer part of $x$).
Thus in the $n$--step mutinomial model which is given by (\ref{2.11}) the number of growth rates is proportional to $n^{1/4}$.
If $U^{(1)}$ takes on a finite number of values, then we can construct the multinomial models with a fixed number of growth rates by
letting $U^{(n,i)}=U^{(i)}$ for any $i,n$. In this case the proof of Theorem \ref{thm2.1} is simpler than for the general case and does not require
Lemma \ref{lem3.2}.
\end{rem}
\section{Auxiliary lemmas}\label{sec3}\setcounter{equation}{0}
We start with a standard result, but since we could not find a direct reference its proof is given here for readers' convenience.
\begin{lem}\label{lem3.0}
Let $n\in\mathbb{N}$ and consider a probability space together with a filtration
${\{\mathcal{G}_k\}}_{k=0}^n$ and a positive $d$--dimensional adapted stochastic process $\{Q^{(k)}\}_{k=0}^n$.
For any $k$ let $Y(k)=\phi_k(Q^{(0)},...,Q^{(k)})$ where $\phi_k:\mathbb{R}^d\times...\mathbb{R}^d\rightarrow \mathbb{R}_{+}$.
Assume that $E\max_{0\leq k\leq n}Y(k)<\infty$ and define
\begin{equation}\label{3.0+}
A=esssup_{\tau\in\mathcal{M}}E(Y(\tau)|\mathcal{G}_0)
\end{equation}
where $\mathcal{M}$ is the set of stopping times with respect to the above filtration
with values in $\{0,1,...,n\}$.
Assume that the sequence ${\{(\frac{Q^{(k)}_1}{Q^{(k-1)}_1},...,\frac{Q^{(k)}_d}{Q^{(k-1)}_d})\}}_{k=1}^n$
is a sequence of i.i.d. random vectors such that for any $k$, $(\frac{Q^{(k)}_1}{Q^{(k-1)}_1},...,\frac{Q^{(k)}_d}{Q^{(k-1)}_d})$
is independent of $\mathcal{G}_{k-1}$. Let $p(\cdot)$ be the probability density function of
$(\frac{Q^{(1)}_1}{Q^{(0)}_1},...,\frac{Q^{(1)}_d}{Q^{(0)}_d})$.
Then
\begin{equation}\label{3.0++}
A=\psi_0(Q^{(0)})
\end{equation}
where $\psi_0$ is given by the following dynamical programming relations
\begin{eqnarray}\label{3.0+++}
&\psi_n=\phi_n \ \ \mbox{and} \ \mbox{for} \ 0\leq k<n, \ \psi_k(x^{(0)},x^{(1)},...,x^{(k)})=\max(\phi_k(x^{(0)},x^{(1)},\\
&...,x^{(k)}),\int_{z\in\mathbb{R}^d}\psi_{k+1}(x^{(0)},x^{(1)},...,x^{(k)},x^{k,z} )p(z_1,...,z_n) dz_1...dz_n)\nonumber
\end{eqnarray}
where $x^{k,z}=(x^{k,z}_1,...,x^{k,z}_d)\in\mathbb{R}^d$ is given by $x^{k,z}_i=x^{(k)}_iz_i$, $i\leq d$.
\end{lem}
\begin{proof}
It is well known (see \cite{PS}) that $A=A(0)$ can be calculated by the following dynamical programming relations
\begin{equation}\label{3.0++++}
\begin{split}
A(n)=Y(n) \  \ \mbox{and} \   \mbox{for} \ 0\leq k<n ,\ A(k)=\max(Y(k),E(A(k+1)|\mathcal{G}_k)).
\end{split}
\end{equation}
By using backward induction, (\ref{3.0+++}) and (\ref{3.0++++}) we obtain
that for any $0\leq k\leq n$
\begin{equation}\label{3.0+++++}
A(k)=\psi_k(Q^{(0)},Q^{(1)},...,Q^{(k)})
\end{equation}
and the result follows.
\end{proof}
Next, we derive several estimates that will be used in this Section. For any $n\in\mathbb{N}$ and $1\leq i\leq d $ denote,
\begin{eqnarray}\label{3.1+}
&S^W_i(t)=S_i(0)\exp((r+\mu_i-{\sum_{j=1}^d \sigma^2_{ij}}/{2})t
 +\sum_{j=1}^d\sigma_{ij} W_j(t)), \\
 &J^{(i)}_n=\max_{0\leq k\leq n-1}\sup_{kT/n\leq t\leq (k+1)T/n}|{S}^W_i(t)-{S}^W_i(kT/n)|,\nonumber\\
 &D^W=\sum_{k=1}^d\sup_{0\leq t\leq T}|S^W_k(t)| \ \mbox{and} \ D=\sum_{k=1}^d\sup_{0\leq t\leq T}|S_k(t)| \nonumber.
\end{eqnarray}
By using the inequality $|\exp(x)-\exp(y)|\leq |x-y|\exp(\max(x,y))$ we obtain that
for any $1\leq i\leq d$
\begin{eqnarray}\label{3.1++}
&J^{(i)}_n\leq D^W\bigg(|r+\mu_i-{\sum_{j=1}^d \sigma^2_{ij}}/{2}|T/n+\\
&\sum_{j=1}^d |\sigma_{ij}|\max_{0\leq k\leq n-1}\sup_{kT/n\leq t\leq (k+1)T/n}|{W}_j(t)-W_j(kT/n)|\bigg).\nonumber
\end{eqnarray}
Fix $1\leq j\leq d$. From the scaling property of the Brownian motion it follows
\begin{equation*}
E\bigg(\max_{0\leq k\leq n-1}\sup_{kT/n\leq t\leq (k+1)T/n}|{W}_j(t)-W_j(kT/n)|^4\bigg)\leq n E(\sup_{0\leq t\leq T/n}|{W}_j(t)|^4)\leq\frac{c_1}{n}
\end{equation*}
for some constant $c_1$. This together with (\ref{3.1++}) and the Holder inequality gives
\begin{equation}\label{3.1+++}
E\sum_{i=1}^d J^{(i)}_n\leq {c_2}n^{-1/4}
\end{equation}
for some constant $c_2$.
Let $\mathcal{T}^{1,n}\subset\mathcal{T}$ be the set of stopping times with values in
$\{0,\frac{T}{n},\frac{2T}{n},...,T\}$.
Set,
\begin{equation}\label{3.2}
V^{(1)}_n=\sup_{\tau\in\mathcal{T}^{1,n}}E(\exp(-r\tau)F(S(\tau),\tau)).
\end{equation}
\begin{lem}\label{lem3.1}
There exists a constant $C_1$ such that for any $n$
\begin{equation}\label{3.3}
0\leq V-V^{(1)}_n\leq C_1 n^{-1/4}.
\end{equation}
\end{lem}
\begin{proof}
The inequality $0\leq V-V^{(1)}_n$ is obvious. Thus it remains to prove that $V-V^{(1)}_n\leq  C_1 n^{-1/4}$.
Fix $n\in\mathbb{N}$ and choose $\epsilon>0$.
There exist $\tau\in\mathcal{T}$ such that
\begin{equation}\label{3.4}
 V<\epsilon+E(\exp(-r\tau)F(S(\tau),\tau)).
\end{equation}
Define the random variable $\sigma=\min\{t\in\{0,\frac{T}{n},\frac{2T}{n},...,T\}| \ t\geq\tau\}$. Observe that
$\sigma\in\mathcal{T}^{1,n}$ and $\tau\leq\sigma\leq\tau+\frac{T}{n}$. Thus from (\ref{2.7+}) it follows
\begin{eqnarray}\label{3.5}
&V-V^{(1)}_n<\epsilon+E(\exp(-r\tau)F(S(\tau),\tau))-E(\exp(-r\sigma)F(S(\sigma),\sigma))\\
&\leq\epsilon+E|F(S(\tau),\tau)-\exp(-rT/n)F(S(\sigma),\sigma)|\nonumber\\
&\leq\epsilon+E|F(S(\tau),\tau)-F(S(\sigma),\sigma)|+\frac{rT}{n}E(F(S(\sigma),\sigma))\leq\epsilon+
\frac{LT}{n}(1+ED)\nonumber\\
&+L E(\sum_{i=1}^d |S_i(\tau)-S_i(\sigma)|)+\frac{rT}{n}(F(S(0),0)+L D+LT(1+D)).\nonumber
\end{eqnarray}
Next, set the event $Q=\{N(\sigma)>N(\tau)\}$. Notice that for any $1\leq i\leq d$
\begin{equation}\label{3.6}
|S_i(\tau)-S_i(\sigma)|\leq  \mathbb{I}_QD+ 2(\sup_{0\leq t\leq T}\prod_{j=1}^{N(t)}(1+U^{(j)}_i)) J^{(i)}_n.
\end{equation}
>From (\ref{3.1+++}), (\ref{3.5})--(\ref{3.6}) and the Cauchy–-Schwarz inequality we obtain that there exist constants $c_3,c_4$ such that
\begin{eqnarray}\label{3.7}
&V-V^{(1)}_n<\epsilon+ c_3n^{-1/4}+c_4 \sqrt{P(Q)}.
\end{eqnarray}
>From the strong Markov property of the Poisson process (with respect to the natural filtration generated by $S$)
and the inequality $\tau\leq \sigma\leq \tau+\frac{T}{n}$ we obtain
\begin{equation}\label{3.8}
P(Q)\leq P(N(T/n)>0)=1-\exp(-\lambda T/n)\leq \frac{\lambda T}{n}.
\end{equation}
The result follows by combining (\ref{3.7})--(\ref{3.8}) and letting $\epsilon\downarrow 0$.
\end{proof}
For any $n$ define ${S}^{1,n}=({S}^{1,n}_1,...,S^{1,n}_d)$ by
\begin{equation}\label{3.9}
\begin{split}
{S}^{1,n}_i(t)=S_i(0)\exp((r+\mu_i-{\sum_{j=1}^d \sigma^2_{ij}}/{2})t
 +\sum_{j=1}^d\sigma_{ij} W_j(t))\prod_{j=1}^{N(t)}(1+U^{n,j}_i), \  1\leq i\leq d.
\end{split}
\end{equation}
Let $\mathcal{T}^{2,n}$ be the set of stopping times with respect to the filtration
$\{\sigma(S(u),S^{1,n}(u)|u\leq t)\}_{t=0}^T$ with values in
$\{0,\frac{T}{n},\frac{2T}{n},...,T\}$.
Set
\begin{equation}\label{3.10}
V^{(2)}_n=\sup_{\tau\in\mathcal{T}^{2,n}}E(\exp(-r\tau)F(S^{1,n}(\tau),\tau)).
\end{equation}
\begin{lem}\label{lem3.2}
There exists a constant $C_2$ such that for any $n$
\begin{equation}\label{3.11}
|V^{(2)}_n-V^{(1)}_n|\leq {C_2} n^{-1/8}.
\end{equation}
\end{lem}
\begin{proof}
>From Lemma \ref{lem3.0} it follows that $V^{(1)}_n=\sup_{\tau\in\mathcal{T}^{2,n}}E(\exp(-r\tau)F(S(\tau),\tau))$.
This together with (\ref{2.7+}) gives
\begin{eqnarray}\label{3.15}
&|V^{(2)}_n-V^{(n)}_1|\leq
L E (\sup_{0\leq t\leq T}\sum_{i=1}^d|{S}^{1,n}_i(t)-S_i(t)|)\leq L\times\\
&E (D^W(\sum_{i=1}^d\sup_{0\leq t\leq T}|\prod_{j=1}^{N(t)}(1+U^{(j)}_i)-\prod_{j=1}^{N(t)}(1+U^{n,j}_i)|))=\nonumber\\
&L E(D^W) E(\sum_{i=1}^d\sup_{0\leq t\leq T}|\prod_{j=1}^{N(t)}(1+U^{(j)}_i)-\prod_{j=1}^{N(t)}(1+U^{n,j}_i)|).\nonumber
\end{eqnarray}
Next, we estimate the last term from (\ref{3.15}). Fix $1\leq i\leq d$. For any $n,j\in\mathbb{N}$ set
$A^{n,j}_i=U^{n,j}_i-U^{(j)}_i$. The sequences $U^{n,j},U^{(j)}$ are independent of the Poisson process $N$. Thus
\begin{eqnarray}\label{3.16}
&E(\sup_{0\leq t\leq T}|\prod_{j=1}^{N(t)}(1+U^{(j)}_i)-\prod_{j=1}^{N(t)}(1+U^{n,j}_i)|)=\\
&\sum_{k=1}^\infty P(N(T)=k)E(\max_{1\leq m\leq k}|\prod_{j=1}^m(1+U^{(j)}_i)-\nonumber\\
&\prod_{j=1}^m(1+U^{(j)}_i+A^{n,j}_i)|)\leq\sum_{k=1}^\infty P(N(T)=k)\times\nonumber\\
&E\bigg(\max_{1\leq m\leq k}\sum_{j=1}^m\sum_{1\leq q_1<q_2<...q_j\leq m}\prod_{s=1}^j |A^{n,q_s}_i|\times\nonumber\\
&\prod_{1\leq s\leq m, s\notin\{q_1,...,q_j\}}(1+U^{(s)}_i)\bigg)\leq\sum_{k=1}^\infty P(N(T)=k)\times\nonumber\\
&E\bigg(\sum_{j=1}^m\sum_{1\leq q_1<q_2<...q_j\leq m}
\prod_{s=1}^j |A^{n,q_s}_i|
\prod_{1\leq s\leq m, s\notin\{q_1,...,q_j\}}(1+U^{(s)}_i)\bigg)=\nonumber\\
&\sum_{k=1}^\infty P(N(T)=k)\sum_{m=1}^k\sum_{j=1}^m
\binom{m}{j}(E|A^{n,1}_i|)^j(1+EU^{(1)}_i)^{m-j}=
\nonumber
\end{eqnarray}
from (\ref{2.8+++}) and the inequality $|x^m-y^m|\leq m|x-y|(\max(x,y))^{m-1}$ (for $x,y\geq 0$, and $m\in\mathbb{N}$) we obtain
\begin{eqnarray}\label{3.17}
&=\sum_{k=1}^\infty P(N(T)=k)\sum_{m=1}^k((1+EU^{(1)}_i+E|A^{n,1}_i|)^m-\\
&(1+EU^{(1)}_i)^m)\leq n^{-1/8} \sum_{k=1}^\infty P(N(T)=k)\sum_{m=1}^k m (2+EU^{(1)}_i)^{m-1}\nonumber\\
&=n^{-1/8}\sum_{m=1}^\infty m P(N(T)\geq m)(2+EU^{(1)}_i)^{m-1}\leq \nonumber\\
&n^{-1/8}\sum_{m=1}^\infty m \frac{(\lambda T)^m}{m!}(2+EU^{(1)}_i)^{m-1}={\lambda T\exp(\lambda T (2+EU^{(1)}_i))}n^{-1/8}.
\nonumber
\end{eqnarray}
By combining (\ref{3.15})--(\ref{3.17}) we obtain (\ref{3.11}).
\end{proof}
For any $n\in\mathbb{N}$, $0\leq k\leq n$ and $1\leq i\leq d$ set
\begin{eqnarray}\label{3.18}
&g^{(i)}_n=\max_{0\leq m\leq n} |(1+(1-\exp(-\lambda T/n))E U^{n,1}_i\big)^{-m}-\exp(\mu_i m T/n)|,\\
&Z^{n,k}=\sum_{j=1}^k\mathbb{I}_{\{N(jT/n)-N((j-1)T/n)\geq 1\}} \ \mbox{and}\nonumber\\
&{S}^{2,n}_i(t)=S_i(0)\exp((r-{\sum_{j=1}^d \sigma^2_{ij}}/{2})kT/n
 +\sum_{j=1}^d\sigma_{ij} W_j(kT/n))\times\nonumber\\
 &\frac{\prod_{j=1}^{Z^{n,k}} (1+U^{n,j}_i)}{\big(1+(1-\exp(-\lambda T/n))E U^{n,1}_i\big)^k}, \ \mbox{for} \ \frac{kT}{n}\leq t<\frac{(k+1)T}{n}.\nonumber
\end{eqnarray}
Notice that $Z^{n,k}\sim N^{n,k}$ (the last term was defined before ({\ref{2.11})).
>From (\ref{2.6}) and
(\ref{2.8+++}) it follows that there exists a constant $c_5$ such that
\begin{equation}\label{3.19}
\begin{split}
g^{(i)}_n \leq \frac{c_5}{n}, \ \ \forall  1\leq i\leq d.
\end{split}
\end{equation}
Denote ${S}^{2,n}=({S}^{2,n}_1,...,S^{2,n}_d)$. Let $\mathcal{T}^{3,n}$ be the set of stopping times with respect to the filtration
$\{\sigma(S(u),S^{1,n}(u),S^{2,n}(u)|u\leq t)\}_{t=0}^T$ with values in
$\{0,\frac{T}{n},\frac{2T}{n},...,T\}$.
Define
\begin{equation}\label{3.20}
V^{(3)}_n=\sup_{\tau\in\mathcal{T}^{3,n}}E(\exp(-r\tau)F(S^{2,n}(\tau),\tau)).
\end{equation}
\begin{lem}\label{lem3.3}
There exists a constant $C_3$ such that for any $n$
\begin{equation}\label{3.21}
|V^{(3)}_n-V^{(2)}_n|\leq C_3 n^{-1/8}.
\end{equation}
\end{lem}
\begin{proof}
Fix $n\in\mathbb{N}$. Introduce the events $Q_k={\{N(kT/n)-N((k-1)T/n)> 1\}}$, $k\leq n$.
Notice that for any $1\leq i\leq d$
\begin{eqnarray}\label{3.22}
&\max_{0\leq k \leq n } |{S}^{2,n}_i(kT/n)-{S}^{1,n}_i(kT/n)| \leq \mathbb{I}_{\{\cup_{k=1}^n Q_k\}}\sup_{0\leq t \leq T }({S}^{1,n}_i(t)+\\
&{S}^{2,n}_i(t))+
G^{(i)}_n\exp(|\mu_iT|)\sup_{0\leq t \leq T} {S}^{1,n}_i(t)\nonumber.
\end{eqnarray}
>From Lemma \ref{lem3.0} it follows that $V^{(2)}_n=\sup_{\tau\in\mathcal{T}^{3,n}}E(\exp(-r\tau)F(S(\tau),\tau))$.
It is easy to verify that the terms  $E(\sup_{0\leq t \leq T } ({S}^{1,n}_i(t)+{S}^{2,n}_i(t))^2)$ are uniformly (with respect to $n$) bounded.
Thus (\ref{3.22}) and the Cauchy–-Schwarz inequality gives that there exist constants $c_6,c_7$ such that
\begin{eqnarray}\label{3.23}
&|V^{(3)}_n-V^{(2)}_n|\leq L E\sum_{i=1}^d \max_{0\leq k \leq n } |{S}^{2,n}_i(kT/n)-{S}^{1,n}_i(kT/n)|\leq \\
&c_6 n^{-1/8}+ c_7 \sqrt{P(\cup_{k=1}^n Q_k)}\leq c_6 n^{-1/8}+c_7 \sqrt{\frac{\lambda^2 T^2}{n}}\nonumber
\end{eqnarray}
and the result follows.
\end{proof}
Next, we derive estimates on a discrete probability spaces. For any $n$ and $0\leq k\leq n$ define
\begin{eqnarray}\label{3.24}
&S^{3,n}_i(t)=S_i(0)\exp(krT/n+\sum_{m=1}^k\sum_{j=1}^d(\sqrt{T/n}\sigma_{ij}\xi^{(m)}_j-\sigma^2_{ij}T/(2n)))\\
&\times\frac{\prod_{j=1}^{N^{n,k}} (1+u^{n,j}_i)}{\big(1+(1-\exp(-\lambda T/n))E_n u^{n,1}_i\big)^k}, \ \ \mbox{for} \ \  kT/n\leq t<(k+1)T/n \nonumber\\
&\mbox{and} \ V^{(4)}_n=\sup_{\tau\in\mathcal{T}_n} E_n(\exp(-\tau T/n)F(\tau,S^{3,n}(\tau)), \ \mbox{where} \
{S}^{3,n}=({S}^{3,n}_1,...,S^{3,n}_d) \nonumber.
\end{eqnarray}
\begin{lem}\label{lem3.4}
There exists a constant $C_4$ such that for any $n$
\begin{equation}\label{3.25}
|V^{(4)}_n-V_n|\leq C_4 n^{-1/2}.
\end{equation}
\end{lem}
\begin{proof}
For any $n\in\mathbb{N}$, $0\leq m\leq n$ and $1\leq i\leq d$ set,
\begin{eqnarray}\label{3.26}
&B^{m,n}_i=\ln(1+\sum_{j=1}^d \sqrt{T/n}
 \sigma_{ij}\xi^{(m)}_j)-\sum_{j=1}^d(\sqrt\frac{T}{n}\sigma_{ij}\xi^{(m)}_j-\sigma^2_{ij}T/(2n)),\\
&C^{m,n}_i=\frac{T}{2n}(\sum_{j=1}^d \sigma^2_{ij}-\sum_{j=1}^d \sum_{k=1}^d \sigma_{ij}\sigma_{ik}\xi^{(m)}_j\xi^{(m)}_k). \nonumber
\end{eqnarray}
Observe that for any $1\leq j\leq d$, $|\xi^{(1)}_j|\leq\sqrt{d+1}$ a.s. Thus
by using the Taylor series of $\ln(1+x)$ we obtain
\begin{equation}\label{3.27}
|B^{m,n}_i-C^{m,n}_i|\leq c_8 n^{-3/2}
\end{equation}
for some constant $c_8$.
>From (\ref{2.9+}) it follows that for any $i,n$ the random variables ${\{C^{m,n}_i\}}_{m=1}^n$ are i.i.d.
with mean $0$. This together with the Doob-Kolmogorov inequality (see \cite{LS}) gives
\begin{equation}\label{3.28}
\begin{split}
E_n(\max_{1\leq m\leq n}|\sum_{k=1}^m C^{k,n}_i|^2)\leq 4 E_n(|\sum_{k=1}^n C^{k,n}_i|^2)\leq 4n E_n(|C^{1,n}_i|^2)\leq \frac{c_{9}}{n}
\end{split}
\end{equation}
for some constant $c_{9}$.
It is easy to verify that the terms $E_n(\sum_{i=1}^d \sup_{0\leq t \leq T}(S^{3,n}_i(t)+S^{(n)}_i(t))^2)$, $n\in\mathbb{N}$ are
uniformly bounded. From (\ref{3.27})--(\ref{3.28}) and the Cauchy–-Schwarz inequality we obtain that there exists a constant $C_4$ such that
\begin{eqnarray*}
&|V^{(4)}_n-V_n|\leq LE_n(\sum_{i=1}^d \sup_{0\leq t\leq T}|S^{3,n}_i(t)-S^{(n)}_i(t)|)\leq \\
&L E_n(\sum_{i=1}^d \max_{1\leq m\leq n}|\sum_{k=1}^m B^{m,n}_i|\sup_{0\leq t \leq T}(S^{3,n}_i(t)+S^{(n)}_i(t)))\leq\\
&L E_n(\sum_{i=1}^d (c_8 n^{-1/2}+\max_{1\leq m\leq n}|\sum_{k=1}^m C^{m,n}_i|)\sup_{0\leq t \leq T}(S^{3,n}_i(t)+S^{(n)}_i(t)))\\
&\leq C_4 n^{-1/2}.
\end{eqnarray*}
\end{proof}
\section{Proof of main results}\label{sec4}\setcounter{equation}{0}
The following result which we state without proof was established in \cite{S} under more general assumptions
(see Theorem 2.1, Corollary 3.1 and Lemma 5.1 there). This result is the main tool that we use in order to complete
the proof of Theorem \ref{thm2.1}.
\begin{thm}\label{thm4.1}
Consider a probability space together with a sequence $X^{(1)},X^{(2)},...$ of i.i.d. $d$--dimensional random vectors and a $d$--dimensional random vector $Y$ . Assume that
$E\sum_{i=1}^ d (|X^{(1)}_i|^3+|Y_i|^3)<\infty$, $EX^{(1)}=EY$ and
\begin{equation}\label{4.1}
\begin{split}
E(X^{(1)}_iX^{(1)}_j)=E(Y_i Y_j) \ \ \forall i,j\in\{1,...,d\}.
\end{split}
\end{equation}
For any $z>0$ it is possible to extend our probability space to $(\tilde{\Omega},\tilde{F},\tilde{P})$ which contain a sequence of i.i.d. random vectors $Y^{(1)},Y^{(2)},...$ such that
$Y^{(1)}\sim Y$,
and for any $n\in\mathbb{N}$
\begin{equation}\label{4.2}
\tilde{P}(\max_{1\leq k\leq n}\sum_{i=1}^d|\sum_{m=1}^k X^{(m)}_i-Y^{(m)}_i|>z)\leq  \frac{Cn}{z^3}\sum_{i=1}^d \tilde{E}(|X^{(1)}_i|^3+|Y_i|^3)
\end{equation}
for some constant $C$ which independent of $X^{(1)},X^{(2)},...$, $Y$ and $z$ ($\tilde{E}$ denotes the expectation with respect to $\tilde{P}$).
Furthermore, for any $k>1$ the random vectors $X^{(1)},...$\\
$,X^{(k-1)},Y^{(k)},Y^{(k+1)},...$ are independent.
\end{thm}
>From Lemmas \ref{lem3.1}--\ref{lem3.4} it follows that there exists a constant $C_5$ such that
for any $n\in\mathbb{N}$, $|V-V_n|\leq C_5 n^{-1/8}+|V^{(2)}_n-V^{(3)}_n|$. Thus in order to complete the proof of Theorem \ref{thm2.1}
it remains to establish the following lemma.
\begin{lem}\label{lem4.1}
For any $\epsilon>0$ there exists a constant $\tilde{C}_{\epsilon}$ such that for any $n$
\begin{equation}\label{4.3}
|V^{(2)}_n-V^{(3)}_n|<\tilde{C}_{\epsilon}n^{\epsilon-\frac{1}{8}}.
\end{equation}
\end{lem}
\begin{proof}
Fix $n\in\mathbb{N}$ and $\epsilon>0$.
We start with proving the inequality
\begin{equation}\label{4.3+}
V^{(2)}_n-V^{(3)}_n<\tilde{C}_{\epsilon}n^{\epsilon-\frac{1}{8}}.
\end{equation}
>From (\ref{2.9+}) and Theorem \ref{thm4.1} it follows that we can construct on the same probability space $(\tilde{\Omega},\tilde{F},\tilde{P})$ two sequence of i.i.d. random vectors
$X^{(1)},...,X^{(n)}$,
$Y^{(1)},...,Y^{(n)}$ such that $X^{(1)}\sim W(\frac{T}{n})$, $Y^{(1)}\sim \sqrt\frac{T}{n}\xi^{(1)}$
and
\begin{equation}\label{4.4}
\tilde{P}(\max_{1\leq k\leq n}\sum_{i=1}^d |\sum_{m=1}^k X^{(m)}_i-Y^{(m)}_i|>n^{-1/8})\leq C_6 n^{-1/8}
\end{equation}
for some constant $C_6$. Furthermore, for any $k>1$ the random vectors $X^{(1)},...,X^{(k-1)},Y^{(k)}$\\
$,...,Y^{(n)}$
are independent.
We can extend the constructed probability space such that it will contain also two independent sequences of i.i.d.
random vectors $F^{(1)},...,F^{(n)}$ and $G^{(1)},...,G^{(n)}$ which are independent of $X^{(1)},...,X^{(n)}, Y^{(1)},...,Y^{(n)}$
and satisfy  $F^{(1)}\sim \rho^{n,1}$,  $G^{(1)}\sim U^{(n,1)}$ ($\rho^{n,1}$, $U^{n,1}$ were defined in Section 2).

For any $p\leq n$ and $i\leq d$ set,
\begin{eqnarray}\label{4.4+}
&Z^{(p)}=\sum_{j=1}^p F^{(j)}, \ D_i=\max_{0\leq k \leq n}\frac{\prod_{j=1}^{Z^{(k)}} (1+G^{(j)}_i)}{\big(1+(1-\exp(-\lambda T/n))\tilde{E}G^{(1)}_i\big)^k},\\
&D^X_i=S_i(0)\exp(rT)\max_{0\leq k\leq n}\exp(\sum_{m=1}^k\sum_{j=1}^d(\sigma_{ij}X^{(m)}_j-\sigma^2_{ij}T/(2n)))\nonumber\\
&\mbox{and} \  D^Y_i=S_i(0)\exp(rT)\max_{0\leq k\leq n}\exp(\sum_{m=1}^k\sum_{j=1}^d(\sigma_{ij}Y^{(m)}_j-\sigma^2_{ij}T/(2n))).\nonumber
\end{eqnarray}
It is easy to verify that there exists a constant $\hat{C}$ and for any $p\geq 1$ there exists a constant $\hat{C}_p$ (the above constants does not depend on $n$) such that
\begin{equation}\label{4.4++}
\begin{split}
\sum_{i=1}^d ED_i<\hat{C} \ \ \mbox{and} \ \ E((\sum_{i=1}^d D^X_i+D^Y_i)^p)<\hat{C}_p.
\end{split}
\end{equation}

Define
$\tilde{S}(t)=(\tilde{S}_1(t),...,\tilde{S}_d(t))$
and  $\hat{S}(t)=(\hat{S}_1(t),...,\hat{S}_d(t))$ by
\begin{eqnarray}\label{4.5}
&\tilde{S}_i(t)=S_i(0)\exp(rkT/n+\sum_{m=1}^k\sum_{j=1}^d(\sigma_{ij}X^{(m)}_j-\sigma^2_{ij}T/(2n)))\times\\
&\frac{\prod_{j=1}^{Z^{(k)}} (1+G^{(j)}_i)}{\big(1+(1-\exp(-\lambda T/n))\tilde{E}G^{(1)}_i\big)^k} \ \ \mbox{and} \ \
 \hat{S}_i(t)=S_i(0)\exp(rkT/n+\nonumber\\
 &\sum_{m=1}^k\sum_{j=1}^d(\sigma_{ij}Y^{(m)}_j-\sigma^2_{ij}T/(2n)))
\frac{\prod_{j=1}^{Z^{(k)}} (1+G^{(j)}_i)}{\big(1+(1-\exp(-\lambda T/n))\tilde{E}G^{(1)}_i\big)^k},\nonumber
\end{eqnarray}
where $1\leq i\leq d$, $0\leq k\leq n$ and $\frac{kT}{n}\leq t<\frac{(k+1)T}{n}$.

Next, let $\tilde{T}$ be the set of stopping times with respect to the filtration  ${\{\sigma(\tilde{S}_i(u)|u\leq t)\}}_{t=0}^T$
with values in $\{0,\frac{T}{n},\frac{2T}{n},...,T\}$ and let $\hat{T}$ be the set of stopping times with respect to the filtration  ${\{\sigma(\tilde{S}(u),\hat{S}(u)|u\leq t)\}}_{t=0}^T$ with values in $\{0,\frac{T}{n},\frac{2T}{n},...,T\}$.
Set,
\begin{eqnarray}\label{4.6}
&\tilde{V}=\sup_{\tau\in\tilde{T}}\tilde{E}(\exp(-r\tau)F(\tilde{S}(\tau),\tau)), \ \
\hat{V}=\sup_{\tau\in\hat{T}}\tilde{E}(\exp(-r\tau)F(\hat{S}(\tau),\tau))\\
&\mbox{and} \ \ Q=\{\max_{1\leq k\leq n}\sum_{i=1}^d|\sum_{m=1}^k X^{(m}_i-Y^{(m)}_i|>n^{-1/8}\}.\nonumber
\end{eqnarray}
Observe that $\tilde{S}\sim S^{2,n}$ and $\hat{S}\sim S^{3,n}$.
>From Lemma \ref{lem3.0} it follows that $\tilde{V}=V^{(2)}_n$ and $\hat{V}=V^{(3)}_n$, in the last equality we used the fact that  $X^{(1)},...,X^{(k-1)},Y^{(k)},...,Y^{(n)}$ are independent for any $k>1$.
Since $\tilde{T}\subset\hat{T}$
then from (\ref{4.4}), (\ref{4.4++})--(\ref{4.6}), the inequality $|\exp(x)-\exp(y)|\leq |x-y|\max(\exp(x),\exp(y))$
and the Holder inequality we obtain
\begin{eqnarray}\label{4.6+}
&V^{(2)}_n-V^{(3)}_n\leq\sup_{\tau\in\hat{T}}\tilde{E}(\exp(-r\tau)F(\tilde{S}(\tau),\tau)-\\
&\sup_{\tau\in\hat{T}}\tilde{E}(\exp(-r\tau)F(\hat{S}(\tau),\tau))\leq L\tilde{E}(\sup_{0\leq t\leq T}\sum_{i=1}^d |\tilde{S}_i(t)-\hat{S}_i(t)|)\nonumber\\
&\leq L \tilde{E} ((n^{-1/8}\max_{1\leq i,j\leq d}|\sigma_{ij}|+\mathbb{I}_Q)(\sum_{i=1}^d D_i(D^X_i+D^Y_i))) \nonumber\\
&\leq L \hat{C}((C_6)^{1/p} n^{-1/8+\epsilon}(\hat{C}_{q})^{1/q}+\max_{1\leq i,j\leq d} |\sigma_{ij}| \hat{C}_1n^{-1/8}) \nonumber
\end{eqnarray}
where $p=\frac{1}{1-8\epsilon}$ and $q=\frac{1}{8\epsilon}$. This completes the proof of the inequality
(\ref{4.3+}). The inequality $V^{(3)}_n-V^{(2)}_n<\tilde{C}_{\epsilon}n^{\epsilon-\frac{1}{8}}$ can be proved in a similar way, just take
 $X^{(1)}\sim \sqrt\frac{T}{n}\xi^{(1)}$ and $Y^{(1)}\sim  W(\frac{T}{n})$.
\end{proof}
\begin{rem}\label{rem4.1}
In  several cases our approach can be extended also
for path dependent payoffs which satisfy
Lipschitz type conditions. Let $M[0,t]$ be the space of Borel
measurable functions $\upsilon=(\upsilon_1,...,\upsilon_d):[0,t]\rightarrow \mathbb{R}^d$
with the uniform metric
$d_{0t}(\upsilon,\tilde{\upsilon})=\sup_{0\leq{s}\leq{t}}\sum_{i=1}^d|\upsilon_i(s)
$\\
$-\tilde{\upsilon}_i(s)|$. For each $t>0$ let $F_t$ be
a nonnegative function on $M[0,t]$ such that for any $t\geq{s}\geq{0}$ and
$\upsilon,\tilde{\upsilon}\in{M[0,t]}$,
\begin{eqnarray}\label{4.7}
&|F_s(\upsilon)-F_s(\tilde{\upsilon})|\leq{Ld_{0s}(\upsilon,\tilde{\upsilon})}
\ \ \mbox{and} \  \
|F_t(\upsilon)-F_s({\upsilon})|\nonumber\\
&\leq{L(|t-s|(1+\sup_{u\in{[0,t]}}\sum_{i=1}^d|\upsilon_i(u)|)
+\sup_{u\in{[s,t]}}\sum_{i=1}^d|\upsilon_i(u)-\upsilon_i(s)|)}.\nonumber
\end{eqnarray}
Consider an American option with the payoff process
$Y(t)=F_t(S)$
where $S=S(\omega)\in{M[0,\infty)}$ is a
random function taking the value $S(\omega,t)$ at
$t\in{[0,\infty)}$. When considering
 $F_t(S^B)$ for $t<\infty$ we take the
restriction of $S$ to the interval $[0,t]$.
The term $V=\sup_{\tau\in\mathcal{T}}E(\exp(-r\tau)Y(\tau))$ gives an arbitrage--free price for our model.
For the multinomial models we consider the payoffs
$Y^{(n)}(k)=F_{\frac{kT}{n}}(S^{(n)})$ and the arbitrage--free prices
$V_n=\sup_{\tau\in\mathcal{T}_n} E_n(\exp(-\tau T/n)Y^{(n)}(\tau))$.
Lemmas \ref{lem3.1}, \ref{lem3.4} and \ref{lem4.1}
can be extended to this setup in a way that does not ruin the estimates of Theorem \ref{thm2.1}. The problem is with Lemmas \ref{lem3.2}--\ref{lem3.3}. For path dependent options the equality before $(\ref{3.15})$ does not follows from Lemma \ref{lem3.1} (although is correct) and the the first inequality in (\ref{3.23}) should be replaced by
$|V^{(3)}_n-V^{(2)}_n|\leq L E\sum_{i=1}^d \sup_{0\leq t \leq T } |{S}^{2,n}_i(t)-{S}^{1,n}_i(t)|.$
The way to fix it is to consider a piecewise constant approximations of $S$, $\dot{S}^{(n)}(t)=S_{kT/n}$ for $kT/n\leq t<(k+1)T/n$, $0\leq k\leq n$. If we could provide an estimates of the term $E\max_{0\leq k\leq n} |F_{\frac{kT}{n}}(S)-F_{\frac{kT}{n}}(\dot{S}^{(n)})|$ then both of the Lemmas \ref{lem3.2}--\ref{lem3.3} could be extended.
Observe that,
$E(\max_{0\leq k\leq n} |F_{\frac{kT}{n}}(S)-F_{\frac{kT}{n}}(\dot{S}^{(n)})|)\leq L E(\sup_{0\leq t\leq T}\sum_{i=1}^d |S_i(t)-\dot{S}^{(n)}_i(t)|).$
For the Black--Scholes model (no Poisson process) the last term was estimated in (\ref{3.1+++}). Thus for the Black--Scholes model the estimates from Theorem \ref{thm2.1}
remain valid for path dependent options which satisfy Lipschitz type conditions.
However, if we allow jumps of compound Poisson type than in general the term $E(\sup_{0\leq t\leq T}\sum_{i=1}^d |S_i(t)-\dot{S}^{(n)}_i(t)|)$
should not tend to $0$. Consider a specific type of path dependent options which are given by $F_t(S)=\max(M,\sup_{0\leq t\leq T}\max_{1\leq i \leq d} S_i(t))$, where $M>0$ is some constant (Russian options). Recall the terms $J^{(i)}_n$ that were defined in (\ref{3.1+}) and the events $Q_k$ which were defined before (\ref{3.22}). For Russian options we have the following inequality
\begin{eqnarray*}
&\max_{0\leq k\leq n} |F_{\frac{kT}{n}}(S)-F_{\frac{kT}{n}}(\dot{S}^{(n)})|\leq L\mathbb{I}_{\{\cup_{j=1}^n Q_j\}}\sup_{0\leq t\leq T}F_t(S)+\\
&\sup_{0\leq t\leq T}\prod_{j=1}^{N(t)}(1+U^{(j)}_i)\sum_{i=1}^d J^{(i)}_n.
\end{eqnarray*}
By using estimates that were derived in Section 3 we obtain that $E(\max_{0\leq k\leq n} |F_{\frac{kT}{n}}(S)-F_{\frac{kT}{n}}(\dot{S}^{(n)})|)$ is of order $n^{-1/4}$ and so the estimates from Theorem \ref{thm2.1}
are also valid for Russian options in Merton's model.
\end{rem}
${}$\\
\textbf{\Large{Acknowledgments:}} \\
\\
I would like to express my deepest gratitude to
my P.hD adviser, Yuri Kifer, for guiding me and helping me to present this work.
I am also very grateful to A.Zaitsev for valuable discussions.
This research was partially supported by ISF grant no. 130/06.}

\end{document}